\pdfoutput=1
\RequirePackage{ifpdf}
\ifpdf 
\documentclass[pdftex]{sigma}
\else
\documentclass{sigma}
\fi

\numberwithin{equation}{section}

\newtheorem{Theorem}{Theorem}[section]
\newtheorem*{Theorem*}{Theorem}

\newtheorem{Lemma}[Theorem]{Lemma}
\newtheorem{Proposition}[Theorem]{Proposition}
 { \theoremstyle{definition}
\newtheorem{Definition}[Theorem]{Definition}

\newtheorem{Example}[Theorem]{Example}
\newtheorem{Remark}[Theorem]{Remark} }

\def\R{\mathbb{R}}
\def\e{\varepsilon}

\newcommand{\la}{\lambda}
\def\N{{\mathbb{N}}}

\def\a{{\alpha}}
\def\b{{\beta}}

\def\id{\mathrm{id}}
\allowdisplaybreaks

\begin{document}

\allowdisplaybreaks

\newcommand{\arXivNumber}{2504.21359}

\renewcommand{\PaperNumber}{084}

\FirstPageHeading

\ShortArticleName{Ultra-Discretization of Yang--Baxter Maps and Independence Preserving Property}

\ArticleName{Ultra-Discretization of Yang--Baxter Maps,\\ Probability Distributions and Independence\\ Preserving Property}

\Author{Hiroki KONDO~$^{\rm a}$, Sachiko NAKAJIMA~$^{\rm b}$ and Makiko SASADA~$^{\rm b}$}

\AuthorNameForHeading{H.~Kondo, S.~Nakajima and M.~Sasada}

\Address{$^{\rm a)}$~Faculty of Data Science, Shimonoseki City University, Yamaguchi 751-8510, Japan}
\EmailD{\href{mailto:kondo-hi@shimonoseki-cu.ac.jp}{kondo-hi@shimonoseki-cu.ac.jp}}

\Address{$^{\rm b)}$~Graduate School of Mathematical Sciences, The University of Tokyo, Tokyo 153-8914, Japan}
\EmailD{\href{mailto:sachiko.nakajima@steam21.com}{sachiko.nakajima@steam21.com}, \href{mailto:sasada@ms.u-tokyo.ac.jp}{sasada@ms.u-tokyo.ac.jp}}

\ArticleDates{Received May 02, 2025, in final form October 07, 2025; Published online October 13, 2025}

\Abstract{We study the relationship between Yang--Baxter maps and the independence preserving (IP) property, motivated by their role in integrable systems, from the perspective of ultra-discretization. Yang--Baxter maps satisfy the set-theoretic Yang--Baxter equation, while the IP property ensures independence of transformed random variables. The relationship between these two seemingly unrelated properties has recently started to be studied by Sasada and Uozumi (2024). Ultra-discretization is a concept primarily used in the context of integrable systems and is an area of active research, serving as a method for exploring the connections between different integrable systems. However, there are few studies on how the stationary distribution for integrable systems changes through ultra-discretization. In~this paper, we introduce the concept of ultra-discretization for probability distributions, and prove that the properties of being a Yang--Baxter map and having the IP property are both preserved under ultra-discretization. Applying this to quadrirational Yang--Baxter maps, we confirm that their ultra-discrete versions retain these properties, yielding new examples of piecewise linear maps having the IP property. We also explore implications of our results for stationary distributions of integrable systems and pose several open questions.}


\Classification{60E05; 62E10; 37K60; 16T25}

\section{Introduction}

\subsection{Background}\label{section1.1}

Recently, research has begun on the connection between Yang--Baxter maps and the independence preserving property~\cite{SU}, which originated from very different fields.

For a bijective function $F\colon \mathcal{X} \times \mathcal{X} \to \mathcal{X} \times \mathcal{X}$ where $\mathcal{X}$ is a set, $F$ is called a Yang--Baxter map if it satisfies the ``set-theoretical'' Yang--Baxter equation
\begin{equation}\label{eq:YB}
F_{12} \circ F_{13} \circ F_{23} = F_{23} \circ F_{13} \circ F_{12},
\end{equation}
where $F_{ij}$ acts on the $i$-th and $j$-th factors of the product $\mathcal{X} \times \mathcal{X} \times \mathcal{X}$.

On the other hand, when $\mathcal{X}$ is a measurable set and $F$ is a measurable bijection, $F$ is said to have the independence preserving property (IP property for short) if there is a quadruplet of (non-Dirac) probability distributions $\mu$, $\nu$, $\tilde{\mu}$, $\tilde{\nu}$ satisfying $F( \mu \times \nu) = \tilde{\mu} \times \tilde{\nu}$. In other words, there exist independent (non-constant) $\mathcal{X}$-valued random variables $X$, $Y$ such that $U$, $V$ are also independent with $(U,V):=F(X,Y)$.

The background of the study on this new connection lies in the equivalence between the existence of independent and identically distributed (i.i.d.) stationary distributions for $(1+1)$-dimensional lattice models, and the IP property of the map that governs the local evolution of these lattice models~\cite{CSirf}. In recent years, the existence of i.i.d.\ stationary distributions for various examples of lattice models that define {\it integrable systems} has been discovered, and it has been suggested that there may be a connection between the integrability and the IP property \cite{CSirf,CS, CSsurvey}. Based on this idea, the authors of~\cite{SU} demonstrated that a class of Yang--Baxter maps known as quadrirational maps possess the IP property.

Furthermore, the paper~\cite{SU} suggests that there are possible extensions of the study on the relationship between Yang--Baxter maps and the IP property, such as ultra-discrete versions and matrix-valued versions. In this paper, the relationship between Yang--Baxter maps and the IP property is examined from the perspective of ultra-discretization.

Ultra-discretization is a concept primarily used in the context of integrable systems, and similar ideas are referred to as tropicalization in algebraic geometry and the zero-temperature limit in statistical mechanics. It is an operation that transforms the algebra of addition and multiplication into the algebra of min (or max) and addition. Although there are slight differences in the choice of signs and other details, these three terms essentially represent the same operation. One of the simplest examples is that the ultra-discretization of the map $F(x,y)=\bigl(x+y,\frac{x}{y}\bigr)$ is~${F_{\star}(x,y)=(\min\{x,y\},x-y)}$.

The study of ultra-discretization in integrable systems has already been actively pursued. For instance, it is well-known that when the discrete KdV equation (which is a discretization of the KdV equation) is ultra-discretized, the box-ball system is obtained \cite{Inoue_2012, TH98}. Moreover, research on the relation between the ultra-discrete integrable systems and tropical geometry is also being conducted (see~\cite{Inoue_2012} and references therein). On the other hand, in statistical mechanics, it is widely known that the zero-temperature limit of polymer models gives rise to first/last passage percolation, and there is a substantial body of research on this topic (see~\cite{CSjsp} and references therein). Especially, research on \textit{stochastic} integrable models has been particularly active in recent years. However, the statistical approach of studying the stationary distribution of discrete \textit{classical} (i.e., deterministic) integrable systems has just begun in recent years, and studies on how the stationary distribution changes through ultra-discretization has still been scarcely conducted. This present paper lays the foundation for such studies in this direction.

\subsection{Results}

Our main result consists of general theorems on ultra-discretization and concrete results concerning the case of quadrirational maps.

First, as general theorems, we show that the properties of being a Yang--Baxter map and having the IP property are both preserved under ultra-discretization in an appropriate sense (Proposition~\ref{prop:YBZL} and Theorem~\ref{thm:ip}) . In particular, we introduce the concept of ultra-discretization for a class of probability distributions and formulate the ultra-discretization of the IP property. Based on the newly introduced definition, for example, the ultra-discretization of the gamma distribution is a shifted exponential distribution (Proposition~\ref{prop:gam}), and the ultra-discretization of the generalized inverse Gaussian distribution is a shifted and truncated exponential distribution (Theorem~\ref{thm:trop-dist}).

As a second result, by applying the previous general theorems, we confirmed that the ultra-discretization of quadrirational Yang--Baxter maps are also Yang--Baxter maps, and further, that they have the IP property (Theorems~\ref{thm:ZT quad maps} and~\ref{thm:ip-quad}). As a result, a large number of examples of piecewise linear functions that having the IP property have been obtained. Such examples were previously very few, such as $F_{\mathrm{Exp}}(x,y):=(\min\{x,y\},x-y)$, which satisfies the IP property with the exponential distributions and geometric distributions. By constructing a general theory, we~are able to systematically provide many such examples. Moreover, in special cases, the result implies that the continuous-valued version of the box-ball system, or its generalization, has an independent and identically distributed (i.i.d.) stationary distribution. The IP property for piecewise linear functions is also related to results concerning the stationary distribution of the zero-temperature limit of polymer models~\cite{CSjsp}.

\subsection{Structure of the rest of this paper}

In Section~\ref{section2}, we formulate the ultra-discretization of rational maps and show that the property of being (parametrized) Yang--Baxter maps is preserved under the ultra-discretization (Proposition~\ref{prop:YBZL}). Then we recall the definition of quadrirational maps and explicitly calculate their ultra-discretization. In Section~\ref{section3}, we formulate the ultra-discretization for a family of probability distributions and show that the independence preservation property is preserved under ultra-discretization in an appropriate sense (Theorem~\ref{thm:ip}). Furthermore, we explicitly compute the ultra-discretization for several specific probability distributions (Theorem~\ref{thm:trop-dist}). Finally, we show that the ultra-discretization of quadrirational maps, obtained in Section~\ref{section2}, has the IP property (Theorem~\ref{thm:ip-quad}). In Section~\ref{section4}, we discuss the relation of our result to various integrable systems and address several open problems. A particularly important open question is the direct relation between being a Yang--Baxter map (or a member of Yang--Baxter maps with parameters) and having the IP property. Combining the results of the paper~\cite{SU} and the current one suggests a~deeper relationship, but a direct mathematical connection has not yet been obtained.

\section{Ultra-discretization of rational maps}\label{section2}

\subsection{Formulation and abstract results}
For $\e\in\R_+\!=\!\{x\in\R\mid x>0\}$,
define $S_{\e}\colon\R\to\R_+$ by $S_{\e}(x)=\exp\bigl(-\e^{-1}x\bigr)$ and ${S_{\e}^n\colon \R^n\to\R_+^n}$~by%
\[
S_{\e}^n(x_1,\ldots,x_n)=(S_{\e}(x_1),\ldots,S_{\e}(x_n)).
\]
These are bijections with $S_{\e}^{-1}(x)=-\e\log x$ and
$(S_{\e}^n)^{-1}(x_1,\ldots,x_n)=\bigl(S_{\e}^{-1}(x_1),\ldots,S_{\e}^{-1}(x_n)\bigr)$ their inverses.
In the following, we simply write $S_{\e}$ instead of $S_{\e}^n$.

Let $\R_+(x_1,\ldots,x_n)$ denote the semi-ring of rational functions of
form \smash{$f=\frac{P}{Q}$}, where $P$ and $Q$ are
polynomials of $x_1,\ldots,x_n$ with positive real coefficients.

For $f\in\R_+(x_1,\ldots,x_n)$, denote by $f_{\star}$ the function on $\R^n$ which is obtained by replacing $(+,\times)$-algebra with $(\min,+)$-algebra, i.e.,
$f\mapsto f_{\star}$ is the semi-ring homomorphism
from $\R_+(x_1,\ldots,x_n)$ to the semi-ring of functions $\R^n\to\R$ endowed with operations $(\min,+)$ which satisfies
\[
c\mapsto0,\qquad x_i\mapsto x_i,\qquad i=1,\ldots,n,
\]
where $c\in\R_+$ is any positive constant function.

\begin{Example}
If we define $f\in\R_+(x,y)$ by $f(x,y)=\frac{ax^k+by^{\ell}}{cx^m+dy^{n}}$ where $a,b,c,d >0$ and ${k,\ell,m,n \in \mathbb N}$, we have
$f_{\star}(x,y)=\min\{kx,\ell y\} - \min \{mx, ny\}$.
\end{Example}

The function~$f_{\star}$ is the {\it ultra-discretization} of $f$ in the following sense.

\begin{Proposition}\label{prop:zero-temp limit}
For $f\in\R_+(x_1,\ldots,x_n)$,
$(S_{\e})^{-1}\circ f\circ S_{\e}\colon \R^n\to\R$ converges
uniformly on $\R_+^n$
to $f_{\star}$ as $\e\downarrow0$.
\end{Proposition}
\begin{proof}
Since we have $(S_{\e})^{-1}\circ f\circ S_{\e}
=(S_{\e})^{-1}\circ P\circ S_{\e}-(S_{\e})^{-1}\circ Q\circ S_{\e}$,
it is sufficient to show that
$(S_{\e})^{-1}\circ P\circ S_{\e}$ converges to $P_{\star}$ when
$P$ is a polynomial of $x_1,\ldots,x_n$ with positive coefficients.

Let \smash{$P=\sum_{j}g_j$}, where each $g_j$ is a monomial with coefficient $c_j>0$.
Then $P_\star=\min_j\{g_{j\star}\}$ and we can write
\begin{align*}
S^{-1}_\varepsilon \circ P \circ S_\varepsilon
&{}= -\varepsilon \log \biggl( \sum_j c_j \exp\bigl(-\varepsilon^{-1} g_{j\star}\bigr) \biggr)\\
&{}= P_\star - \varepsilon \log \biggl( c_{j_0} + \sum_{j\neq j_0} c_j \exp\bigl(-\varepsilon^{-1}(g_{j\star}-P_\star)\bigr) \biggr),
\end{align*}
where the index $j_0=j_0(x_1,\ldots,x_n)$ with $P_\star(x_1,\ldots,x_n)=g_{j_0\star}(x_1,\ldots,x_n)$.
Since $g_{j\star}-P_\star\geq 0$ and $c_j$ are positive, we obtain
\[
\min_j c_j \;\le\; c_{j_0}+\sum_{j\neq j_0} c_j \exp\bigl(-\varepsilon^{-1}(g_{j\star}-P_\star)\bigr) \le \sum_j c_j,
\]
which implies the uniform convergence on 
$\mathbb{R}_+^n$.
\end{proof}

We will only use the uniform convergence on compact subsets of $\R_+^n$.
Note that the ultra-discretization here is also called the tropicalization or the zero-temperature limit in other literature.

\begin{Remark}
Note that the map $f\mapsto f_{\star}$ is not injective even if we identify polynomials that differ only in their coefficients.
For example, if we define $f,g\in\R_+(x,y)$ by
$f(x,y)=x^2+xy+y^2$ and $g(x,y)=x^2+y^2$, we have
$f_{\star}(x,y)=\min\{2x,x+y,2y\}$ and
$g_{\star}(x,y)=\min\{2x,2y\}$, which are the same as
functions on $\R^2$.
\end{Remark}

\begin{Remark}
If we use $S_{-\e}$ instead of $S_{\e}$,
we will have the image in the $(\max,+)$ algebra as the limit
in Proposition~\ref{prop:zero-temp limit}.
\end{Remark}

Let $F\colon \mathcal{X}\times \mathcal{X}\to \mathcal{X}\times \mathcal{X}$ be a
Yang--Baxter map on a set $\mathcal{X}$, i.e., $F$ satisfies the Yang--Baxter equation (\ref{eq:YB}) in Section~\ref{section1.1}.
If $S\colon \mathcal{X}\to \mathcal{X}$ is a bijection and we denote
$S\colon \mathcal{X}\times \mathcal{X}\to \mathcal{X}\times \mathcal{X}$ by the map which acts
as $S$ on each component, then clearly
$S^{-1}\circ F\circ S$ is also a Yang--Baxter map.

Combining this with Proposition~\ref{prop:zero-temp limit},
we have the following.

\begin{Proposition}
If $f_1,f_2\in\R_+(x,y)$ and
$F=(f_1,f_2)\colon\R_+^2\to\R_+^2$ is a Yang--Baxter map,
then $F_{\star}=(f_{1\star},f_{2\star})$ is a Yang--Baxter map
on $\R^2$.
\end{Proposition}

Next, we also consider the ultra-discretization of Yang--Baxter maps with parameters.
Let ${F^{\alpha,\beta}\colon \mathcal{X}\times \mathcal{X}\to \mathcal{X}\times \mathcal{X}}$ be a family of maps with parameters $\alpha$, $\beta$.
We say that $\bigl(F^{\alpha,\beta}\bigr)$ satisfies the Yang--Baxter property if the Yang--Baxter equation
\[
F_{12}^{\alpha,\beta}\circ F_{13}^{\alpha,\gamma}\circ F_{23}^{\beta,\gamma}=F_{23}^{\beta,\gamma}\circ F_{13}^{\alpha,\gamma}\circ F_{12}^{\alpha,\beta}
\]
holds for any $(\alpha,\beta,\gamma)$.
If \smash{$\bigl(F^{\alpha,\beta}\bigr)$} also satisfies
\[
F_{21}^{\beta,\alpha}\circ F_{12}^{\alpha,\beta}=\id,
\]
\smash{$\bigl(F^{\alpha,\beta}\bigr)$} is called a {\it reversible} Yang--Baxter family.

We can prove the Yang--Baxter property for the ultra-discretization with parameters in the same way as the theorem above.

\begin{Proposition}\label{prop:YBZL}
If $f_1^{\alpha,\beta},f_2^{\alpha,\beta}\in\R_+(x,y,\alpha,\beta)$ and
the family of maps
\smash{$\bigl(F^{\alpha,\beta}\bigr)_{\alpha,\beta\in\R_+}$}, where
\smash{$F^{\alpha,\beta}=\bigl(f_1^{\alpha,\beta},f_2^{\alpha,\beta}\bigr)\colon\R_+^2\to\R_+^2$},
satisfies the Yang--Baxter property,
then \[
\bigl(F_{\star}^{\alpha,\beta}\bigr)_{\alpha,\beta\in\R}=\bigl(f_{1\star}^{\alpha,\beta}, f_{2\star}^{\alpha,\beta}\bigr)_{\alpha,\beta\in\R}
\]
 satisfies the Yang--Baxter property on $\R^2$. Moreover, if $\bigl(F^{\alpha,\beta}\bigr)$ is a reversible Yang--Baxter family, then \smash{$\bigl(F_{\star}^{\alpha,\beta}\bigr)$} is also a reversible Yang--Baxter family.
\end{Proposition}

\begin{Remark}
In the preceding proposition, the quantities $\alpha$, $\beta$ were temporarily treated as variables of $F$ and reinterpreted as parameters after ultra-discretization.
Accordingly, their domain in $F_{\star}$ has changed from $\mathbb{R}_{+}$ to $\mathbb{R}$.

Note that if we treat $f_1^{\alpha,\beta}$, $f_2^{\alpha,\beta}$ as functions in $\R_+(x,y)$ of fixed parameters $\a$, $\b$ and perform ultra-discretization, the information about the parameters can be lost, and we cannot obtain the ultra-discretized Yang--Baxter maps with parameters.
\end{Remark}

\begin{Remark}
The strategy to use the ultra-discretization to obtain piecewise linear Yang--Baxter maps was also studied in~\cite{KNW2, KNW} in the context of discrete integrable systems.
\end{Remark}

\subsection{Quadrirational maps}
In this subsection, we review the quadrirational maps
studied in~\cite{ABS,PSTV,SU},
and consider the ultra-discretization of them.

A map $F\colon\mathbb{CP}^1\times\mathbb{CP}^1\to\mathbb{CP}^1\times\mathbb{CP}^1; (x,y)\mapsto(u,v)$ is said to be
{\it quadrirational} if both $F$ and its companion map
$\overline{F}\colon\mathbb{CP}^1\times\mathbb{CP}^1\to\mathbb{CP}^1\times\mathbb{CP}^1; (x,v)\mapsto(u,y)$ are birational, where $\mathbb{CP}^1$ is the one-dimensional complex projective space.
In~\cite{ABS}, it is shown that
any quadrirational map $F(x,y)=\bigl(u(x,y),v(x,y)\bigr)$ has
the form
\[
u(x,y)=\frac{a(y)x+b(y)}{c(y)x+d(y)},\qquad
v(x,y)=\frac{A(x)y+B(x)}{C(x)y+D(x)},
\]
where $a$, $b$, $c$, $d$, $A$, $B$, $C$, $D$ are polynomials whose degrees are at most two.
The authors of~\cite{ABS} also showed that,
any quadrirational map
such that the highest degree of the coefficients in
the representations of both $u$ and $v$ is 2
is M\"{o}bius equivalent to any of the five families of maps
$F_{\mathrm{I}}=\bigl(F_{\mathrm{I}}^{\alpha,\beta}\bigr),
F_{\mathrm{II}}=\bigl(F_{\mathrm{II}}^{\alpha,\beta}\bigr), \ldots,
F_{\mathrm{V}}=\bigl(F_{\mathrm{V}}^{\alpha,\beta}\bigr)$, where $\alpha$ and $\beta$ are
complex parameters.

Although the representative maps $F_{\mathrm{I}}, \ldots, F_{\mathrm{V}}$ are reversible Yang--Baxter maps, not all quadrirational maps satisfy
the Yang--Baxter property because M\"{o}bius actions do not preserve the Yang--Baxter property.

In~\cite{PSTV}, it is shown that there are additional 5
families of quadrirational Yang--Baxter maps up to
{\it YB equivalence}, named $H_{\mathrm{I}}$, $H_{\mathrm{II}}$, $H_{\mathrm{III}}^A$,
$H_{\mathrm{III}}^B$ and $H_{\mathrm{V}}$, and all of them are reversible.
Here two families of quadrirational maps
$\bigl(F^{\alpha,\beta}\bigr)$ and \smash{$\bigl(\tilde{F}^{\alpha,\beta}\bigr)$} are
YB equivalent if there exists a family of bijection
$\phi(\alpha)$ in $\mathbb{CP}^1$ such that
$\tilde{F}^{\alpha,\beta}=
\bigl(\phi(\alpha)^{-1}\times\phi(\beta)^{-1}\bigr)\circ F^{\alpha,\beta}\circ(\phi(\alpha)\times\phi(\beta))$ holds.

It is also shown in~\cite{PSTV} that $H_{\mathrm{I}}$ and $H_{\mathrm{II}}$ have subtraction-free representatives named $H_{\mathrm{I}}^+$ and~$H_{\mathrm{II}}^+$ respectively,
and $H_{\mathrm{III}}^A$ and $H_{\mathrm{III}}^B$ are originally subtraction-free:
\begin{align*}
&H_{\mathrm{I}}^{+,\alpha,\beta}(x,y)=\biggl(\frac{y}{\alpha}\frac{\beta+\alpha x+\beta y+\alpha\beta xy}{1+x+y+\beta xy},\frac{x}{\beta}\frac{\alpha+\alpha x+\beta y+\alpha\beta xy}{1+x+y+\alpha xy}\biggr),\\
&H_{\mathrm{II}}^{+,\alpha,\beta}(x,y)=\biggl(\frac{y}{\alpha}\frac{\beta+\alpha x+\beta y}{1+x+y},\frac{x}{\beta}\frac{\alpha+\alpha x+\beta y}{1+x+y}\biggr),\\
&H_{\mathrm{III}}^{A,\alpha,\beta}(x,y)=\biggl(\frac{y}{\alpha}\frac{\alpha x+\beta y}{x+y},
\frac{x}{\beta}\frac{\alpha x+\beta y}{x+y}\biggr),\\
&H_{\mathrm{III}}^{B,\alpha,\beta}(x,y)=\biggl(y\frac{1+\beta xy}{1+\alpha xy},x\frac{1+\alpha xy}{1+\beta xy}\biggr).
\end{align*}

In what follows, we restrict the variables $x$, $y$ and
the parameters $\alpha$, $\beta$ of these maps to lie in~$\R_+$
for applying ultra-discretization.

Define bijections $I,\theta_{\alpha}\colon\R_+\to\R_+$ by $I(x)=\frac1x$ and
$\theta_{\alpha}(x)=\alpha x$ for $\alpha\in\R_+$.
As mentioned in~\cite{SU},
$H_{\mathrm{III}}^A$ and $H_{\mathrm{III}}^B$ have the following relationship:
\[
H_{\mathrm{III}}^{B,\alpha,\beta}=((I\circ\theta_{\alpha})\times\id)
\circ H_{\mathrm{III}}^{A,\alpha,\beta}\circ(\id\times(I\circ\theta_{\beta})),
\]
in particular they are M\"{o}bius equivalent.

Note that we can consider naturally the limit of \smash{$H_{\mathrm{III}}^{B,\alpha,\beta}$} when
$\alpha$ or $\beta$ tend to 0, while it is not possible for \smash{$H_{\mathrm{III}}^{A,\alpha,\beta}$}.
We will introduce Yang--Baxter families
with similar property
which are M\"{o}bius equivalent to \smash{$H_{\mathrm{I}}^{+,\alpha,\beta}$} and
\smash{$H_{\mathrm{II}}^{+,\alpha,\beta}$}.

We define \smash{$G_{\mathrm{I}}=\bigl(G_{\mathrm{I}}^{\alpha,\beta}\bigr)$} by
\[
G_{\mathrm{I}}^{\alpha,\beta}=((I\circ\theta_{\alpha})\times\id)
\circ H_{\mathrm{I}}^{+,\alpha,\beta}\circ(\id\times(I\circ\theta_{\beta})).
\]
The map \smash{\smash{$G_{\mathrm{I}}^{\alpha,\beta}$}} can be written explicitly as
\[
G_{\mathrm{I}}^{\alpha,\beta}(x,y)
=\biggl(y\frac{1+\beta x+\beta y+\beta xy}{1+\alpha x+\beta y+\alpha xy},x\frac{1+\alpha x+\alpha y+\alpha xy}{1+\alpha x+\beta y+\beta xy}\biggr).
\]
We also define \smash{$G_{\mathrm{II}}=\bigl(G_{\mathrm{II}}^{\alpha,\beta}\bigr)$} by
\[
G_{\mathrm{II}}^{\alpha,\beta}
=(\theta_{\alpha^{-1}}\times\theta_{\beta^{-1}})\circ
H_{\mathrm{II}}^{+,\alpha^{-1},\beta^{-1}}\circ(\theta_{\alpha}\times\theta_{\beta}).
\]
The map \smash{$G_{\mathrm{II}}^{\alpha,\beta}$} can be written explicitly as
\[
G_{\mathrm{II}}^{\alpha,\beta}(x,y)
=\biggl(y\frac{1+\beta x+\beta y}{1+\alpha x+\beta y},x\frac{1+\alpha x+\alpha y}{1+\alpha x+\beta y}\biggr).
\]

\begin{Proposition}
Maps $G_{\mathrm{I}}$ and $G_{\mathrm{II}}$ satisfy the Yang--Baxter property, and both are reversible.
\end{Proposition}
\begin{proof}
The map $G_{\mathrm{II}}$ is a
(reparametrization of a) Yang--Baxter equivalent family of $H_{\mathrm{II}}^+$,
thus~$G_{\mathrm{II}}$ clearly satisfies the Yang--Baxter property.

For $G_{\mathrm{I}}$, note that $H_{\mathrm{I}}^+$ is a reversible Yang--Baxter family
\big(as well as $H_{\mathrm{III}}^A$\big).
We also have by direct calculation that
\[
((I\circ\theta_{\alpha})\times(I\circ\theta_{\beta}))\circ
H_{\mathrm{I}}^{+,\alpha,\beta}\circ
((I\circ\theta_{\alpha})\times(I\circ\theta_{\beta}))
=H_{\mathrm{I}}^{+,\alpha,\beta}.
\]
Now the result follows from~\cite[Proposition~2]{PSTV}.
The reversibility also follows by the same argument.
\end{proof}

\begin{Remark}
For $G_{\mathrm{II}}$, similar construction as $G_{\mathrm{I}}$ does not work,
namely, the family
$\bigl(((I\circ\theta_{\alpha})\times\id)
\circ H_{\mathrm{II}}^{+,\alpha,\beta}\circ(\id\times(I\circ\theta_{\beta}))
\bigr)$ does {\it not} satisfy the Yang--Baxter property.
In particular, the same argument as the proposition above cannot be
applied to $H_{\mathrm{II}}^+$ because
\[
((I\circ\theta_{\alpha})\times(I\circ\theta_{\beta}))\circ
H_{\mathrm{II}}^{+,\alpha,\beta}\circ
((I\circ\theta_{\alpha})\times(I\circ\theta_{\beta}))
\ne H_{\mathrm{II}}^{+,\alpha,\beta}.
\]
\end{Remark}

\begin{Remark}
The map \smash{$G_{\mathrm{II}}^{\alpha,\beta}$} is the same function as
$\psi_{\alpha,\beta}$ in~\cite{KW2}, where all the distributions satisfying the IP property for \smash{$G_{\mathrm{II}}^{\alpha,\beta}$} are characterized. Also, in~\cite{KK}, a family of non-commutative Yang--Baxter maps $\mathcal{K}_{a,b,-c}$ was introduced, and the map \smash{$G_{\mathrm{I}}^{\alpha,\beta}$}
is YB equivalent to the commutative
version of $\mathcal{K}_{1,1,-1}$. Similarly, \smash{$G_{\mathrm{II}}^{\alpha,\beta}$} is YB equivalent to the map $\mathcal{K}_{0,1,-1}$.
\end{Remark}

We can apply Proposition~\ref{prop:YBZL} to
the subtraction-free Yang--Baxter families
$H_{\mathrm{I}}^+$, $G_{\mathrm{I}}$,
$H_{\mathrm{II}}^+$, $G_{\mathrm{II}}$,~$H_{\mathrm{III}}^A$ and $H_{\mathrm{III}}^B$:

\begin{Theorem}\label{thm:ZT quad maps}
Let \smash{$H_{\mathrm{I},\star}^+$}, \smash{$G_{\mathrm{I},\star}$},
\smash{$H_{\mathrm{II},\star}^+$}, \smash{$G_{\mathrm{II},\star}$}, \smash{$H_{\mathrm{III},\star}^A$} and \smash{$H_{\mathrm{III},\star}^B$} be
the ultra-discretization of
\smash{$H_{\mathrm{I}}^+$}, \smash{$G_{\mathrm{I}}$},
\smash{$H_{\mathrm{II}}^+$}, \smash{$G_{\mathrm{II}}$}, \smash{$H_{\mathrm{III}}^A$} and \smash{$H_{\mathrm{III}}^B$} respectively.
Then these families satisfy the Yang--Baxter property on $\mathbb{R}^2$. Moreover, they are reversible.
\end{Theorem}

The maps in Theorem~\ref{thm:ZT quad maps} can be
written as follows:
\begin{gather*}
H_{\mathrm{I},\star}^{+,\alpha,\beta}(x,y)=(y-\alpha+\min\{\beta,\alpha+x,\beta+y,\alpha+\beta+x+y\}-\min\{0,x,y,\beta+x+y\},\\
\hphantom{H_{\mathrm{I},\star}^{+,\alpha,\beta}(x,y)=(}{}
 x-\beta+\min\{\alpha,\alpha+x,\beta+y,\alpha+\beta+x+y\}-\min\{0,x,y,\alpha+x+y\}),\\
G_{\mathrm{I},\star}^{\alpha,\beta}(x,y)=(y+\min\{0,\beta+x,\beta+y,\beta+x+y\}-\min\{0,\alpha+x,\beta+y,\alpha+x+y\},\\
\hphantom{G_{\mathrm{I},\star}^{\alpha,\beta}(x,y)=(}{}
 x+\min\{0,\alpha+x,\alpha+y,\alpha+x+y\}-\min\{0,\alpha+x,\beta+y,\beta+x+y\}),\\
H_{\mathrm{II},\star}^{+,\alpha,\beta}(x,y)=(y-\alpha+\min\{\beta,\alpha+x,\beta+y\}-\min\{0,x,y\},\\
\hphantom{H_{\mathrm{II},\star}^{+,\alpha,\beta}(x,y)=(}{}
 x-\beta+\min\{\alpha,\alpha+x,\beta+y\}
-\min\{0,x,y\}),\\
G_{\mathrm{II},\star}^{\alpha,\beta}(x,y)=(y+\min\{0,\beta+x,\beta+y\}
-\min\{0,\alpha+x,\beta+y\},\\
\hphantom{G_{\mathrm{II},\star}^{\alpha,\beta}(x,y)=(}{}
 x+\min\{0,\alpha+x,\alpha+y\}-\min\{0,\alpha+x,\beta+y\}),\\
H_{\mathrm{III},\star}^{A,\alpha,\beta}(x,y)=(y-\alpha+\min\{\alpha+x,\beta+y\}-\min\{x,y\},\\
\hphantom{H_{\mathrm{III},\star}^{A,\alpha,\beta}(x,y)=(}{}
 x-\beta+\min\{\alpha+x,\beta+y\}-\min\{x,y\}),\\
H_{\mathrm{III},\star}^{B,\alpha,\beta}(x,y)=(y+\min\{0,\beta+x+y\}-\min\{0,\alpha+x+y\},\\
\hphantom{H_{\mathrm{III},\star}^{B,\alpha,\beta}(x,y)=(}{}
 x+\min\{0,\alpha+x+y\}-\min\{0,\beta+x+y\}).
\end{gather*}
These functions are piecewise linear and the Yang--Baxter property of them could be proved directly, but
our framework provides an alternative method for verifying the Yang--Baxter property for these maps.

\begin{Remark}
For $G_{\mathrm{I},\star}$, $G_{\mathrm{II},\star}$ and $H_{\mathrm{III},\star}^{B}$, we can naturally consider the limit when $\a$ or $\b$ tend to $\infty$. Hence, similarly as summarized in the figure at~\cite[p.~4]{SU} for the subtraction-free rational maps, we can consider special cases with~$\a$ or $\b=\infty$, which are ultra-discretization of special cases of subtraction-free rational maps with~$\a$ or $\b=0$. This limiting procedure is also applicable for the IP property discussed in the latter sections.
\end{Remark}

\begin{Remark}
If we use $S_{-\varepsilon}$ instead of $S_{\varepsilon}$
(in other words, $\min$ is replaced by $\max$),
some different Yang--Baxter families are obtained.
However, it can be shown that
the Yang--Baxter families obtained in this way from
$H_{\mathrm{I}}^+$ and $G_{\mathrm{I}}$ are the same as $G_{\mathrm{I},\star}$ and $H_{\mathrm{I},\star}^+$, respectively.
\end{Remark}

\begin{Remark}\label{rem:bbs}
The map \smash{$H_{\mathrm{III}}^{B,\a,\b}$} is known to define the local dynamics of the modified discrete KdV equation \cite{PTV} and it is well known that up to the change of signs of variables, \smash{$H_{\mathrm{III},\star}^{B,\a,\b}$} define the ultra-discrete KdV equation \cite{KNW2, KNW}. In particular, when $\a,\b \in \N \cup \{\infty\}$ and we restrict the possible values of the variables so that
$x \in \{0,1,\dots,\a\}$ and $y \in \{0,1,\dots,\b\}$, then the dynamics associated with the map \smash{$H_{\mathrm{III},\star}^{B,\a,\b}$} (up to the signs again) corresponds to the box-ball system with box capacity $\a$ and carrier capacity $\b$.
\end{Remark}

\section[Ultra-discretization of probability distributions and the independence preserving property]{Ultra-discretization of probability distributions\\ and the independence preserving property}\label{section3}

\subsection{Formulation and abstract results}

In this subsection, we introduce an ultra-discretization of probability distributions, which is new as far as we know. Then we prove that the IP property is conserved under the ultra-discretization in a proper sense.

\begin{Definition}
 Let $(\mu_{\e})_{\e}$ be probability distributions on $\R_+$ indexed by $\e >0$. When $S_{\e}^{-1}(\mu_{\e})$, the pushforward measure of $\mu_{\e}$ by $S_{\e}^{-1}$, converges weakly to a probability measure $\mu$ on $\R$, then we say that $\mu$ is the {\it ultra-discretization} of $(\mu_{\e})_{\e}$. 
 \end{Definition}

Next, we prove a general result concerning the weak convergence of pushforward of weakly converging probability distributions by compactly converging functions.

\begin{Lemma}\label{lem:convergence}
Let $m, \ell \in \N$ and $D \subset \R^m$ be an open set. Let $(\mu_n)_{n \in \mathbb{N}}$ be a sequence of probability distributions on $D$ and $(f_n)_{n \in \N}$ be a sequence of continuous functions from $D$ to $\R^{\ell}$. Assume that there exist a probability distribution $\mu$ on $D$ and a continuous function $f\colon D \to \R^{\ell}$ such~that
\[
\mu_n \to \mu \ \text{weakly}, \qquad f_n \to f \ \text{uniformly on any compact set} \ K \subset D
\]
as $n$ goes to $\infty$. Then we have $\lim_{n \to \infty} f_n (\mu_n) = f (\mu)$ weakly.
\end{Lemma}
\begin{proof}
We prove that for any bounded continuous function $g\colon \R^{\ell} \to \R$,
\[
\lim_{n \to \infty} E[g(f_n(X_n))]=E[g(f(X))]
\]
if $X_n \sim \mu_n$ and $X \sim \mu$. Since $\mu_n \to \mu$ weakly, we have
\[
\lim_{n \to \infty} E[g(f(X_n))]=E[g(f(X))].
\]
Hence, we only need to prove that
\begin{equation}\label{eq:conv}
\lim_{n \to \infty} | E[g(f(X_n))]- E[g(f_n(X_n))] | = 0.
\end{equation}
Now, we prove that $g\circ f_n \to g \circ f$ uniformly on any compact set $K \subset D$. Fix a compact set $K \subset D$ and suppose that $g\circ f_n |_K$ does not converge uniformly to $g \circ f |_K$. Then there exist $\delta>0$, a sequence $(x_k)_{k \in \N} \subset K$ and an increasing sequence $n_k \to \infty \ (k \to \infty)$ such that
\[
|g(f_{n_k}(x_k))-g(f(x_k))| > \delta
\]
for any $k \in \N$. Since $K$ is compact, we can even assume that $x_k$ converges to some $x_{\infty} \in K$ by passing to a subsequence if necessary. Then, since $f_n$ converges to $f$ uniformly on $K$, $ \lim_{k \to \infty} f_{n_k}(x_k) = f(x_{\infty})$ and so
\[
\lim_{k \to \infty} (g(f_{n_k}(x_k))-g(f(x_k)) )= g(f(x_{\infty}))-g(f(x_{\infty}))=0
\]
gives the contradiction.

Having this uniform convergence $g\circ f_n \to g \circ f$ on compact sets, it is straightforward to prove~\eqref{eq:conv}. In fact, assuming $\|g\|_{\infty} \neq 0$, for any $\delta>0$, since $\mu_n \to \mu$ weakly, there exists a~compact set $K_{\delta} \subset D$ such that \smash{$\sup_n P(X_n \notin K_{\delta}) < \frac{\delta}{4\|g\|_{\infty}}$}. Also, there exists $n_{\delta}$ such that for any $n \ge n_{\delta}$ and $x \in K_{\delta}$, \smash{$|g(f_n(x))-g(f(x))| < \frac{\delta}{2}$}. Hence,
\begin{align*}
 | E[g(f_n(X_n))]- E[g(f(X_n))] | &\le E[|g(f_n(X_n))-g(f(X_n))|] \\
 & \le E\bigl[2\|g\|_{\infty}\mathbf{1}_{X_n \notin K_{\delta}}\bigr] + E\bigl[|g(f_n(X_n))-g(f(X_n))|\mathbf{1}_{X_n \in K_{\delta}}\bigr] \\
 & < \frac{\delta}{2} + \frac{\delta}{2} =\delta.
\end{align*}
Hence, the lemma follows.
\end{proof}

\begin{Remark}
 The above lemma is essentially the same as~\cite[Lemma 4.3]{SU}. In this paper, we~only need to consider the case where $\ell = m = 2$ and $D = \mathbb{R}^2$, while in~\cite{SU}, the case $D = \mathbb{R}_+^2$ is considered, with the codomains of the functions $f$ and $f_n$ also being $\mathbb{R}_+^2$. Although the proof is essentially the same, we provide it here as well for the sake of completeness.
\end{Remark}

The next theorem is the main result of this subsection, which states that the IP property is inherited to the ultra-discretization if such limits exist both for functions and probability distributions.

\begin{Theorem}\label{thm:ip}
 For a family of continuous functions $F_{\e}\colon \R_+^2 \to \R_+^2$ and probability distributions $\mu_{\e}$, $\nu_{\e}$, $\tilde{\mu}_{\e}$, $\tilde{\nu}_{\e}$ on $\R_+$ indexed by $\e>0$, suppose that $F_{\e}(\mu_{\e} \times \nu_{\e})=\tilde{\mu}_{\e} \times \tilde{\nu}_{\e}$ holds for any sufficiently small $\e >0$. If \smash{$S_{\e}^{-1} \circ F_{\e} \circ S_{\e}$} converges to $F \colon \R^2 \to \R^2$ uniformly on any compact set and $\mu$, $\nu$, $\tilde{\mu}$, $\tilde{\nu}$ are
 the ultra-discretization of
 $(\mu_{\e})_{\e}$, $(\nu_{\e})_{\e}$, $(\tilde{\mu}_{\e})_{\e}$, $(\tilde{\nu}_{\e})_{\e}$, respectively, then ${F(\mu \times \nu)=\tilde{\mu} \times \tilde{\nu}}$.
\end{Theorem}
\begin{proof}
 Let $\bar{F}_{\e}:= S_{\e}^{-1} \circ F_{\e} \circ S_{\e}$. By assumption, \smash{$\bar{F}_{\e} \bigl(S_{\e}^{-1}(\mu_{\e}) \times S_{\e}^{-1}(\nu_{\e})\bigr)= S_{\e}^{-1}(\tilde{\mu}_{\e}) \times S_{\e}^{-1}(\tilde{\nu}_{\e})$}. Since $\mu$, $\nu$, $\tilde{\mu}$, $\tilde{\nu}$ are the ultra-discretization of $\mu_{\e}$, $\nu_{\e}$, $\tilde{\mu}_{\e}$, $\tilde{\nu}_{\e}$, respectively, \smash{$S_{\e}^{-1}(\mu_{\e}) \times S_{\e}^{-1}(\nu_{\e})$} converges weakly to $\mu \times \nu$ and \smash{$S_{\e}^{-1}(\tilde{\mu}_{\e}) \times S_{\e}^{-1}(\tilde{\nu}_{\e})$} converges weakly to $\tilde{\mu} \times \tilde{\nu}$. Then, by Lemma~\ref{lem:convergence}, \smash{$\bar{F}_{\e} \bigl(S_{\e}^{-1}(\mu_{\e}) \times S_{\e}^{-1}(\nu_{\e})\bigr)$} converges weakly to $F(\mu \times \nu)$, and it is also the weak limit of $S_{\e}^{-1}(\tilde{\mu}_{\e}) \times S_{\e}^{-1}(\tilde{\nu}_{\e})$, hence we conclude $F(\mu \times \nu)=\tilde{\mu} \times \tilde{\nu}$.
\end{proof}

\subsection{Examples of probability distributions}
In this subsection, we give concrete examples of
ultra-discretization of probability distributions.

First, as a simple example for the ultra-discretization of probability distributions, we show that the shifted exponential distribution is the ultra-discretization of the sequence of gamma distributions. To state this result, we introduce gamma distributions and shifted exponential distributions.

{\bf Gamma distribution.} For $\lambda, a>0$, the \emph{gamma} distribution with parameters $(\lambda,a)$, which we denote $\mathrm{Ga}(\lambda,a)$, has density
\[\frac{1}{Z} x^{\lambda-1} {\rm e}^{-ax}\mathbf{1}_{x >0}, \]
where $Z$ is a normalizing constant.

{\bf Shifted exponential distribution.} For $\lambda>0$, $a \in \R$, the \emph{shifted exponential} distribution with parameters $(\lambda,a)$, which we denote $\mathrm{sExp} (\lambda,a)$, has density
\[\frac{1}{Z} {\rm e}^{-\la x} \mathbf{1}_{x > a}, \]
where $Z$ is a normalizing constant.

\begin{Proposition}\label{prop:gam}
 Let $\mu_{\e}$ be the probability distribution $\mathrm{Ga}(\la \e,S_{\e}(a))$ for $\la >0,a \in \R$. Then $\mathrm{sExp}(\la, -a)$ is the ultra-discretization of $(\mu_{\e})_{\e}$.
\end{Proposition}

{\samepage

\begin{proof}
 Let $\mu_{\e}:=\mathrm{Ga}(\la \e,S_{\e}(a))$. Then, since the normalizing constant of $\mathrm{Ga}(\lambda,a)$ is \smash{$\frac{\Gamma(\la)}{a^{\la}}$}, the density of $S_{\e}^{-1}(\mu_{\e})$ is
 \[
\frac{{\rm e}^{-\la a}}{\Gamma(\la \e)\e}\exp(-\la x)\exp \biggl(-\exp\biggl(-\frac{x+a}{\e}\biggr)\biggr).
 \]
 Since \smash{$ \lim_{\e \to 0} \Gamma(\la \e) \e =\frac{1}{\la}$} and \smash{$ \lim_{\e \to 0}\exp \bigl(-\exp\bigl(-\frac{x+a}{\e}\bigr)\bigr)=\mathbf{1}_{x > -a} + {\rm e}^{-1} \mathbf{1}_{x=-a}$}, as $\e \to 0$, the above density function converges almost everywhere to
 \[
 \la {\rm e}^{-\la a}\exp(-\la x)\mathbf{1}_{x > -a},
 \]
 which is the density function of $\mathrm{sExp}(\la, -a)$. By Scheff\'e's lemma, we conclude the proof.
\end{proof}

We symbolically denote this situation as $\mathrm{Ga}_{\star}=\mathrm{sExp}$.}

Next, we consider three classes of probability distributions on $\R_+$ with two positive parameters $p, q >0$, which are introduced in~\cite{SU}. In that paper, unnecessary extra parameters have been introduced for the Kummer distribution of Type 2 and the generalized inverse Gaussian distribution, but they were deliberately used to provide a unified notation for how the three distributions relate concerning scale transformation and weak convergence. In the present paper, we keep the same parametrization as the use of such a parametrization makes the statement of main theorem unified and is also the key to understanding ultra-discretization of them.

{\bf Generalized beta prime distribution $\boldsymbol{(p, q)}$.} For $\lambda, a,b \in \R$, $-b < \frac{\lambda}{2} < a$, the \emph{generalized beta prime} distribution with parameters $(\lambda,a,b ; p, q)$, which we denote $\mathrm{gBe}'(\lambda,a,b ; p, q)$, has density
\[\frac{1}{Z} x^{\lambda-1} (1+p x)^{-a-\frac{\lambda}{2}} \bigl(1+q x^{-1}\bigr)^{-b+\frac{\lambda}{2}}\mathbf{1}_{x >0},\]
where $Z$ is a normalizing constant.

{\bf Kummer distribution of Type 2 $\boldsymbol{(p, q)}$.} For $\lambda,b \in \R$, $a >0$, $-b < \frac{\lambda}{2}$, the \emph{Kummer distribution of Type $2$} with parameters $(\lambda,a,b ; p, q)$, which we denote $\mathrm{K} (\lambda,a,b ; p, q)$, has density
\[\frac{1}{Z} x^{\lambda-1}{\rm e}^{-a p x} \bigl(1+q x^{-1}\bigr)^{-b+\frac{\lambda}{2}}\mathbf{1}_{x >0},\]
where $Z$ is a normalizing constant.

{\bf Generalized inverse Gaussian distribution $\boldsymbol{(p, q)}$.} For $\lambda \in \R$, $a,b >0$, the \emph{generalized inverse Gaussian} distribution with parameters $(\lambda,a,b ; p, q)$, which we denote $\mathrm{GIG} (\lambda,a,b ; p, q)$, has density
\[\frac{1}{Z}x^{\lambda-1}{\rm e}^{-a p x-bq x^{-1}}\mathbf{1}_{x >0},\]
where $Z$ is a normalizing constant.

The following three classes of probability distributions are
the ultra-discretization of three classes of probability distributions above with a suitable choice of parameters, which is proved in Theorem~\ref{thm:trop-dist}. In the following, $p, q \in \R$.

{\bf Mixed exponential distribution of beta type $\boldsymbol{(p, q)}$.} For $\lambda, a,b \in \R$, $-b < \frac{\lambda}{2} < a$, the \emph{mixed exponential distribution of beta type} with parameters $(\lambda,a,b ; p, q)$, which we denote $\mathrm{mExpB} (\lambda,a,b ; p, q)$, has density
\[\frac{1}{Z}{\rm e}^{-\lambda x} \bigl({\rm e}^{(a+\frac{\la}{2})(x+p)}\mathbf{1}_{x < -p}+\mathbf{1}_{x > -p}\bigr) \bigl(\mathbf{1}_{x < q}+{\rm e}^{-(b-\frac{\la}{2})(x-q)}\mathbf{1}_{x > q} \bigr) ,\]
where $Z$ is a normalizing constant.

{\bf Mixed exponential distribution of Kummer type $\boldsymbol{(p, q)}$.} For $\lambda,a,b \in \R$, $-b < \frac{\lambda}{2}$, the \emph{mixed exponential distribution of Kummer type} with parameters $(\lambda,a,b ; p, q)$, which we denote $\mathrm{mExpK} (\lambda,a,b ; p, q)$, has density
\[\frac{1}{Z}{\rm e}^{-\lambda x} \mathbf{1}_{x > -p-a} \bigl(\mathbf{1}_{x < q}+{\rm e}^{-(b-\frac{\la}{2})(x-q)}\mathbf{1}_{x > q} \bigr) ,\]
where $Z$ is a normalizing constant.

{\bf Mixed exponential distribution of GIG type $\boldsymbol{(p, q)}$.} For $\lambda,a,b \in \R$ with $a+b+p+q>0$, the \emph{mixed exponential distribution of GIG type} with parameters $(\lambda,a,b ; p, q)$, which we denote $\mathrm{mExpGIG} (\lambda,a,b ; p, q)$, has density
\[\frac{1}{Z}{\rm e}^{-\lambda x} \mathbf{1}_{x > -p-a} \mathbf{1}_{x <b+ q} ,\]
where $Z$ is a normalizing constant.

\begin{Remark}
 The mixed exponential distribution of GIG type is a shifted and truncated exponential distribution.
\end{Remark}

\begin{Theorem}\label{thm:trop-dist}
Fix $p,q \in \mathbb R$.
\begin{enumerate}
\item[$(i)$] Let $\mu_{\e}$ be the probability distribution $\mathrm{Be}'(\la \e, a\e ,b \e; S_{\e}(p), S_{\e}(q))$ for $\lambda, a,b \in \R$, ${-b\! <\! \frac{\lambda}{2}\! <\! a}$. Then $\mathrm{mExpB}(\la , a,b ; p, q)$ is the ultra-discretization of $(\mu_{\e})_{\e}$.

\item[$(ii)$] Let $\mu_{\e}$ be the probability distribution $\mathrm{K}(\la \e, S_{\e}(a) ,b \e; S_{\e}(p), S_{\e}(q))$ for $\lambda,a,b \in \R$, $-b < \frac{\lambda}{2}$. Then $\mathrm{mExpK}(\la , a,b ; p, q)$ is the ultra-discretization of $(\mu_{\e})_{\e}$.

\item[$(iii)$] Let $\mu_{\e}$ be the probability distribution $\mathrm{GIG}(\la \e, S_{\e}(a) ,S_{\e}(b); S_{\e}(p), S_{\e}(q))$ for $\lambda,a,b \in \R$. Assume $a+b+p+q>0$. Then $\mathrm{mExpGIG} (\lambda,a,b ; p, q)$ is the ultra-discretization of $(\mu_{\e})_{\e}$.
\end{enumerate}
\end{Theorem}
\begin{proof}
(i) Let $\mu_{\e}=\mathrm{gBe}'(\la \e, a\e ,b \e; S_{\e}(p), S_{\e}(q))$. Then the density of $S_{\e}^{-1}(\mu_{\e})$ is
\[
\frac{1}{Z_{\e}}\exp(-\la x) \biggl(1 +\exp\biggl(-\frac{p+x}{\e}\biggr)\biggr)^{-\e(a+\frac{\la}{2})}\biggl(1 +\exp\biggl(-\frac{q-x}{\e}\biggr)\biggr)^{-\e(b-\frac{\la}{2})}
 \]
 where
 \[
 Z_{\e}=\int_{\R}\exp(-\la x) \biggl(1 +\exp\bigg(-\frac{p+x}{\e}\biggr)\biggr)^{-\e(a+\frac{\la}{2})}\biggl(1 +\exp\biggl(-\frac{q-x}{\e}\biggr)\biggr)^{-\e(b-\frac{\la}{2})} {\rm d}x.
 \]
 First, by direct computations, we have
 \[
 \lim_{\e \to 0}\biggl(1 +\exp\biggl(-\frac{p+x}{\e}\biggr)\biggr)^{-\e(a+\frac{\la}{2})}= \mathbf{1}_{x+p \ge 0} + {\rm e}^{ (p+x)(a+\frac{\la}{2})}\mathbf{1}_{x+p < 0}
 \]
 and
 \[
 \lim_{\e \to 0}\biggl(1 +\exp\biggl(-\frac{q-x}{\e}\biggr)\biggr)^{-\e(b-\frac{\la}{2})}= \mathbf{1}_{q-x \ge 0} + {\rm e}^{ (q-x)(b-\frac{\la}{2})}\mathbf{1}_{q-x < 0}.
 \]
 Hence, we only need to prove that $ \lim_{\e \to 0}Z_{\e}=Z$, where
 \[
 Z=\int_{\R} \exp(-\la x)\bigl( \mathbf{1}_{x+p \ge 0} + {\rm e}^{ (p+x)(a+\frac{\la}{2})}\mathbf{1}_{x+p < 0}
 \bigr)\bigl( \mathbf{1}_{q-x \ge 0} + {\rm e}^{ (q-x)(b-\frac{\la}{2})}\mathbf{1}_{q-x < 0}\bigr).
 \]
 Since
 \[
\biggl(1 +\exp\biggl(-\frac{p+x}{\e}\biggr)\biggr)^{-\e(a+\frac{\la}{2})} \le \max \bigl\{1, 2^{-\e(a+\frac{\la}{2})} \bigr\} \bigl( \mathbf{1}_{x+p \ge 0} + {\rm e}^{ (p+x)(a+\frac{\la}{2})}\mathbf{1}_{x+p < 0}
 \bigr),
 \]
 and
 \[
\biggl(1 +\exp\biggl(-\frac{q-x}{\e}\biggr)\biggr)^{-\e(b-\frac{\la}{2})} \le \max \bigl\{1, 2^{-\e(b-\frac{\la}{2})} \bigr\} \bigl( \mathbf{1}_{q-x \ge 0} + {\rm e}^{ (q-x)(b-\frac{\la}{2})}\mathbf{1}_{q-x < 0}
 \bigr),
 \]
 there exists $C>0$ such that for sufficiently small $\e>0$,
 \begin{align*}
 &\exp(-\la x) \biggl(1 +\exp\biggl(-\frac{p+x}{\e}\biggr)\biggr)^{-\e(a+\frac{\la}{2})}\biggl(1 +\exp\biggl(-\frac{q-x}{\e}\biggr)\biggr)^{-\e(b-\frac{\la}{2})} \\
 &\qquad{} \le C \exp(-\la x)\bigl( \mathbf{1}_{x+p \ge 0} + {\rm e}^{ (p+x)(a+\frac{\la}{2})}\mathbf{1}_{x+p < 0}
 \bigr)\bigl( \mathbf{1}_{q-x \ge 0} + {\rm e}^{ (q-x)(b-\frac{\la}{2})}\mathbf{1}_{q-x < 0}
 \bigr).
 \end{align*}
 Since the last expression is integrable over $\R$, by the Lebesgue convergence theorem, we conclude $\lim_{\e \to 0}Z_{\e}=Z$.

 (ii)
 Let $\mu_{\e}=\mathrm{K}(\la \e, S_{\e}(a) ,b \e; S_{\e}(p), S_{\e}(q))$. Then the density of $S_{\e}^{-1}(\mu_{\e})$ is
\[
\frac{1}{Z_{\e}}\exp(-\la x) \exp\biggl(-\exp\biggl(-\frac{a+p+x}{\e}\biggr)\biggr)\biggl(1 +\exp\biggl(-\frac{q-x}{\e}\biggr)\biggr)^{-\e(b-\frac{\la}{2})},
 \]
 where $Z_{\e}$ is the normalizing constant. By direct computations,
 \[
 \lim_{\e \to 0}\exp\biggl(-\exp\biggl(-\frac{a+p+x}{\e}\biggr)\biggr)= \mathbf{1}_{x+p+a > 0} + {\rm e}^{-1}\mathbf{1}_{x+p+a = 0}.
 \]
 Hence, by the same argument as (i), the result follows.

 (iii) The proof is completed in the same way as in (i) and (ii).
 \end{proof}

Based on the result of this theorem, from now on, we denote $\mathrm{gBe}'_{\star}=\mathrm{mExpB}$, $\mathrm{K}_{\star}=\mathrm{mExpK}$ and $\mathrm{GIG}_{\star}=\mathrm{mExpGIG}$. We note that $\mathrm{gBe}'(\la,a,b;1,1)$ is the beta prime distribution $\mathrm{Be}'$ and $\mathrm{mExpB}(\la,a,b;0,0)$ is the asymmetric Laplace distribution $\mathrm{AL}$, so as a special case we have the relation $\mathrm{Be}'_{\star}=\mathrm{AL}$.

\subsection{IP property for ultra-discretization of quadrirational maps}

In this subsection, we apply Theorem~\ref{thm:ip} to concrete rational functions and obtain the IP property for the ultra-discretization of rational functions.

First, as a very classical example, recall that for $F_{\mathrm{Ga}}(x,y):=\bigl(x+y, \frac{x}{y}\bigr)$, its ultra-discretization is explicitly computed as $F_{\mathrm{Ga},\star}(x,y)=(\min\{x,y\}, x-y) (=F_{\mathrm{Exp}}(x,y))$. Then, applying Theorem~\ref{thm:ip} and the fact that $\mathrm{Ga}_{\star}=\mathrm{sExp}$ and $\mathrm{Be}'_{\star}=\mathrm{AL}$, from the IP property of gamma distributions for $F_{\mathrm{Ga}}$, we obtain the IP property of shifted exponential distributions for $F_{\mathrm{Ga},\star}$ without any direct computation (see also~\cite[Proposition 5.7]{CSjsp}).

In the following, we prove the IP property for the ultra-discretization of quadrirational maps using the same approach.

Recall that
\smash{$H^{+,\a,\b}_{\mathrm{I},\star}$}, \smash{$H^{+,\a,\b}_{\mathrm{II},\star}$} and
\smash{$H^{A,\a,\b}_{\mathrm{III},\star}$} are obtained in Theorem~\ref{thm:ZT quad maps}. To obtain the IP property of these maps, we first recall the IP property for the quadrirational maps obtained in~\cite{SU}.

\begin{Theorem}[{\cite[Theorem 1.1]{SU}}] \label{thm:ip-quad-su}
Let $\a, \b >0$. If $X$ and $Y$ are $\R^+$-valued independent random variables with the following marginal distributions, then the $\R^+$-valued random variables $U$, $V$ given by $(U,V)=F(X,Y)$ for each map are independent and have the following marginal distributions:
\begin{enumerate}\itemsep=0pt
\item[$(i)$] For \smash{$F= H^{+,\a,\b}_{\mathrm{I}}$},
\begin{alignat*}{3}
&X \sim \mathrm{gBe}' (\la, a, b; \a, 1),\qquad && Y \sim \mathrm{gBe}' (-\la, a, b; \b, 1),& \\
&U \sim \mathrm{gBe}' (-\la, a, b; \a, 1),\qquad && V \sim \mathrm{gBe}' (\la, a, b; \b, 1),&
\end{alignat*}
where $\lambda \in \R$, $a,b >0$, $- \min\{a,b\} < \frac{\lambda}{2} < \min\{a,b\}$.

\item[$(ii)$] For \smash{$F= H^{+,\a,\b}_{\mathrm{II}}$},
\begin{alignat*}{3}
&X \sim \mathrm{K} (\la, a, b; \a, 1), \qquad && Y \sim \mathrm{K} (-\la, a, b; \b, 1),& \\
&U \sim \mathrm{K} (-\la, a, b; \a, 1), \qquad && V \sim \mathrm{K} (\la, a, b; \b, 1),&
\end{alignat*}
where $\lambda \in \R$, $a,b >0$, $-b < \frac{\lambda}{2} < b$.

\item[$(iii)$] For \smash{$F= H^{A,\a,\b}_{\mathrm{III}}$},
\begin{alignat*}{3}
&X \sim \mathrm{GIG} (\la, a, b; \a, 1), \qquad && Y \sim \mathrm{GIG} (-\la, a, b; \b, 1),& \\
&U \sim \mathrm{GIG}(-\la, a, b; \a, 1), \qquad && V \sim \mathrm{GIG} (\la, a, b; \b, 1),&
\end{alignat*}
where $\lambda \in \R$, $a,b >0$.
\end{enumerate}
\end{Theorem}

Then, as a consequence of Theorems~\ref{thm:ip},~\ref{thm:trop-dist} and~\ref{thm:ip-quad-su}, we have the following IP property for the ultra-discretization of quadrirational maps.

\begin{Theorem}\label{thm:ip-quad}
Let $\a,\b \in \R$. If $X$ and $Y$ are $\R$-valued independent random variables with the following marginal distributions, then the $\R$-valued random variables $U$, $V$ given by $(U,V)=F(X,Y)$ for each map are independent and have the following marginal distributions:
\begin{enumerate}\itemsep=0pt
\item[$(i)$] For \smash{$F= H^{+,\a,\b}_{\mathrm{I},\star}$},
\begin{alignat*}{3}
&X \sim \mathrm{gBe}'_{\star} (\la, a, b; \a, 0), \qquad && Y \sim \mathrm{gBe}'_{\star} (-\la, a, b; \b, 0), &\\
&U \sim \mathrm{gBe}'_{\star} (-\la, a, b; \a, 0), \qquad && V \sim \mathrm{gBe}'_{\star} (\la, a, b; \b, 0),&
\end{alignat*}
where $\lambda \in \R$, $a,b >0$, $- \min\{a,b\} < \frac{\lambda}{2} < \min\{a,b\}$.

\item[$(ii)$] For \smash{$F= H^{+,\a,\b}_{\mathrm{II},\star}$},
\begin{alignat*}{3}
&X \sim \mathrm{K}_{\star} (\la, a, b; \a, 0), \qquad && Y \sim \mathrm{K}_{\star} (-\la, a, b; \b, 0), &\\
&U \sim \mathrm{K}_{\star} (-\la, a, b; \a, 0), \qquad && V \sim \mathrm{K}_{\star} (\la, a, b; \b, 0),&
\end{alignat*}
where $a,\lambda \in \R$, $b >0$, $-b < \frac{\lambda}{2} < b$.

\item[$(iii)$] For \smash{$F= H^{A,\a,\b}_{\mathrm{III},\star}$},
\begin{alignat*}{3}
&X \sim \mathrm{GIG}_{\star} (\la, a, b; \a, 0), \qquad && Y \sim \mathrm{GIG}_{\star} (-\la, a, b; \b, 0),& \\
&U \sim \mathrm{GIG}_{\star}(-\la, a, b; \a, 0), \qquad && V \sim \mathrm{GIG}_{\star} (\la, a, b; \b, 0),&
\end{alignat*}
where $a,b,\lambda \in \R$, $a+b > \max\{-\a,-\b\}$.
\end{enumerate}
\end{Theorem}
\begin{proof}
The proof follows by applying Proposition~\ref{prop:zero-temp limit}, Theorems~\ref{thm:ip},~\ref{thm:trop-dist} and~\ref{thm:ip-quad-su} for each quadrirational map and quadruplets of probability distributions.
\end{proof}

Similar to the result of Yang--Baxter property for the ultra-discretization of quadrirational maps, it is not straightforward to obtain this theorem by direct calculations.
Therefore, applying the general theorem for the ultra-discretization of the IP property (namely, Theorem~\ref{thm:ip}) is a~reasonable approach.

\begin{Remark}
 As discussed in Remark~\ref{rem:bbs}, for $\a,\b \in \N \cup\{\infty\}$, it is known that the map~\smash{$H^{B,\a,\b}_{\mathrm{III},\star}$} induces the generalized box-ball system BBS$(\a,\b)$ \cite{TM97}. Hence, to understand the stationary distributions of BBS$(\a,\b)$, its IP property with respect to probability distributions supported on the discrete sets $\{0,1,2,\dots, \a \}$ for $X$ and $\{0,1,2,\dots, \b \}$ for $Y$ has been studied in detail in~\cite{CS}. Also, for general $\a,\b \in \R$,~\cite{CSirf} studied the IP property of \smash{$H^{B,\a,\b}_{\mathrm{III},\star}$} and obtained partial characterization results.
\end{Remark}

\begin{Remark}
Similar to \smash{$H^{B,\a,\b}_{\mathrm{III},\star}$}, for \smash{$H^{+,\a,\b}_{\mathrm{I},\star}$}, \smash{$H^{+,\a,\b}_{\mathrm{II},\star}$} and \smash{$H^{A,\a,\b}_{\mathrm{III},\star}$}, there should be {\it discrete} probability distributions which also satisfy the IP property. We have a guess that the discretized version of probability distributions $\mathrm{Be}'_{\star}$, $\mathrm{K}_{\star}$ and $\mathrm{GIG}_{\star}$ may satisfy the IP property, but we do not pursue it here.
\end{Remark}

\begin{Remark}
The IP property with the similar distribution as Theorem~\ref{thm:ip-quad} also holds for \smash{$G_{\mathrm{I},\star}$}, \smash{$G_{\mathrm{II},\star}$} and $H^{B}_{\mathrm{III},\star}$ since they are obtained by simple change of variables from \smash{$H^{+}_{\mathrm{I},\star}$}, \smash{$H^{+}_{\mathrm{II},\star}$} and~\smash{$H^{A}_{\mathrm{III},\star}$}.
\end{Remark}

\section{Discussion}\label{section4}

\subsection{Relation to integrable systems}

In this subsection, we review relations between the Yang--Baxter maps studied in this paper and some integrable lattice systems, whose dynamics are deterministic or random. Importantly, their (special class of) stationary distributions are often characterized by the IP property of the Yang--Baxter maps.

For deterministic models, the map \smash{$H^{B,\alpha,\beta}_{\mathrm{III}}$} with special parameters $\alpha=1$, $\beta=0$ is known to define the discrete KdV equation, and its ultra-discretization defines the box-ball system (cf.\ \cite{Inoue_2012,KNW, TS90,TH98}). More generally, as mentioned in Remark~\ref{rem:bbs}, \smash{$H^{B,\alpha,\beta}_{\mathrm{III}}$} is related to the modified discrete KdV equation and its ultra-discretization \smash{$H^{B,\a,\b}_{\mathrm{III},\star}$} defines the box-ball system with the box capacity $\a$ and the carrier capacity $\b$ for $\a,\b \in \N \cup \{\infty\}$. On the other hand, in~\cite{CSirf}, under a very general setting, the i.i.d.\ type stationary distributions associated to the deterministic lattice dynamics defined by a local map $F$ are shown to be characterized by the distributions having the IP property for this map $F$. Applying this general result, the i.i.d.\ type stationary distributions of the discrete KdV equation, the modified discrete KdV equation and the box-ball systems with/without capacity are partially characterized in~\cite{CS} and~\cite{CSirf} with the result on the IP property shown in~\cite{LW2}. Other than \smash{$H^{B,\a,\b}_{\mathrm{III},\star}$}, whether the ultra-discretization of quadrirational maps correspond to some (ultra-discrete) integrable systems is, at least to our knowledge, unknown.

For stochastic models, the maps $F_{\mathrm{Ga}}$ and \smash{$F_{\mathrm{Be}}(x,y):=\bigl(\frac{1-y}{1-xy},1-xy \bigr)$} are found to have a~close relation to positive temperature $(1+1)$-dimensional random polymers \cite{CN}. Actually, in~\cite{CN}, the authors characterize all $(1+1)$-dimensional stationary random polymers having a nice integrable property. For this, the IP property of $F_{\mathrm{Ga}}$ and $F_{\mathrm{Be}}$ and characterization of the distributions having the IP property play essential roles. Following the approach of~\cite{CN,CSjsp} studies the zero-temperature version of $(1+1)$-dimensional stationary random polymers, which are related to~$F_{\mathrm{Ga},\star}$ and~$F_{{\mathrm{Be}}',\star}$ where~\smash{$F_{{\mathrm{Be}}'}(x,y)=\bigl(\frac{1+x+y}{xy},\frac{1+y}{x} \bigr)$} is the subtraction-free version of $F_{\mathrm{Be}}$. Same for the positive temperature case, the distributions having the IP property for $F_{\mathrm{Ga},\star}$ and $F_{{\mathrm{Be}}',\star}$ induce the stationary distributions for random polymers, but the complete characterization of the zero-temperature version of $(1+1)$-dimensional stationary random polymers is still missing.

\subsection{Open problems}

In this subsection, we mention some open problems.

First, as the most fundamental question, the mathematical relation between the Yang--Baxter property and the IP property is completely open. Although several similarities between the two properties have been revealed in the previous study~\cite{SU} and in this paper, a direct mathematical connection between them remains to be established, and would be of great interest. In particular, the Yang--Baxter property considered here is a notion pertaining to a family of maps, whereas the IP property is defined for individual maps. Furthermore, the IP property is invariant under arbitrary changes of variables in each component, while the Yang--Baxter property does not share this invariance. These distinctions highlight the conceptual differences between the two, and so a precise formalization of their relationship is a challenging open question.

Next, as a problem related only to the Yang--Baxter property, can we classify all the ultra-discrete quadrirational maps having the Yang--Baxter property, which has been already done for the quadrirational maps \cite{PSTV}? Here, the ultra-discrete quadrirational map is a piecewise linear map from $\R^2$ to $\R^2$ whose inverse map is (well-defined and) also a piecewise linear map, and whose companion map as well as its inverse map are also (well-defined and) piecewise linear.

The \textit{characterization} of the distributions having the IP property for the ultra-discretization of quadrirational maps introduced in this paper is almost open. To the best of our knowledge, the only case that has been fully characterized is $F_{\mathrm{Exp}}=F_{\mathrm{Ga},\star}$, which was done in~\cite{Cr}. In~fact, in~general, for the ultra-discrete versions, there are not only continuous distributions with density functions, but also discrete distributions that have the IP property. In particular, for the three functions given in our main theorem, there should also be discrete distributions having the IP property, but this has not yet been done except for \smash{$H^{B,\a,\b}_{\mathrm{III},\star}$}, which was studied in~\cite{CSirf} and~\cite{CS}. Furthermore, simply finding examples of discrete distributions is not sufficient for a complete characterization, and a complete characterization without any condition on the property of the distribution is considered to be a rather difficult task. We note that for the original quadrirational maps $H_{\mathrm{I}}$, $H_{\mathrm{II}}$ and $H_{\mathrm{III}}$ in Theorem~\ref{thm:ip-quad-su}, the characterization problem has been completely resolved in the following order: the case of $H_{\mathrm{III}}$ in~\cite{LW2}, $H_{\mathrm{II}}$ in~\cite{KW2}, and $H_{\mathrm{I}}$ in~\cite{KLPW2025}.

Related to the last subsection, it is also interesting to see whether there is a deterministic/stochastic lattice dynamics induced by the map \smash{$H^{+}_{\mathrm{I}}$}, \smash{$H^{+}_{\mathrm{I},\star}$}, \smash{$H^{+}_{\mathrm{II}}$} and \smash{$H^{+}_{\mathrm{II},\star}$}.

Given that the quadrirational maps \smash{$H_{\mathrm{I}}$}, \smash{$H_{\mathrm{II}}$} and \smash{$H_{\mathrm{III}}$} are derived from a fully geometric background in~\cite{ABS, PSTV}, one would expect that a similar geometric characterization would be possible for their ultra-discrete versions. In particular, whether the Yang--Baxter property for~$H_{\mathrm{I},\star}$,~$H_{\mathrm{II},\star}$ and~$H_{\mathrm{III},\star}$ can be understood by tropical geometry is an important open question. Moreover, it would be exciting if the IP property can be also understood geometrically in a~certain sense.

Finally, a more practical problem would be the following. In this paper, we have taken the approach of considering an ultra-discrete version of what was known about the relationship between Yang--Baxter property and the IP property. As mentioned in~\cite{SU}, there could be an~approach to consider a higher dimensional version, especially a matrix version, instead of the ultra-discrete version. In fact, part of this work is in progress.

\subsection*{Acknowledgements}
The authors would like to thank the anonymous referees for their valuable comments and suggestions, which helped to improve the quality and clarity of this paper.


\pdfbookmark[1]{References}{ref}
\LastPageEnding

\end{document}